\documentclass{easychair}
\usepackage{microtype}
\usepackage[T1]{fontenc}
\usepackage{ae}
\usepackage{aecompl}
\usepackage[utf8]{inputenc}
\usepackage{url}
\usepackage{amsmath}
\usepackage{todonotes}
\usepackage{amsfonts}
\usepackage{amsthm}
\usepackage{amssymb}
\usepackage{todonotes}
\usepackage{float}
\usepackage{graphicx}
\usepackage{paralist}
\usepackage{rotating}

\usepackage{booktabs}

\usepackage[linesnumbered,ruled]{algorithm2e}

\newtheorem{definition}{Definition}
\newtheorem{lemma}{Lemma}
\newtheorem{theorem}{Theorem}
\newtheorem{remarkb}{Remark}
\newtheorem{proposition}{Proposition}
\newtheorem{example}{Example}
\newtheorem{invariant}{Invariant}
\newtheorem{datastructure}{Data Structure}

\newcommand{\N}{\mathbb{N}}

\newcommand{\lift}{\mathit{lift}}
\newcommand{\Lift}{\mathit{Lift}}

\newcommand{\CP}{\mathit{CPre}}
\newcommand{\Pre}{\mathit{Pre}}

\newcommand{\mina}{\mathit{min}}

\newcommand{\best}{\mathit{best}}

\newcommand{\+}{\mathit{+}}
\newcommand{\z}{{\bar{z}}}
\renewcommand{\S}{\mathbf{S}}
\newcommand{\GG}{\Gamma}

\newcommand{\E}{\ensuremath{\mathcal{E}}\xspace}
\renewcommand{\O}{\ensuremath{\mathcal{O}}\xspace}

\newcommand{\os}{symbolic one-step }

\newcommand{\trho}{\widetilde{\rho}}

\newcommand{\W}{\mathcal{W}}
\newcommand{\Out}{\mathit{Out}}
\newcommand{\In}{\mathit{In}}
\newcommand{\PG}{\mathcal{P}}
\newcommand{\win}{W}

\newcommand{\set}[1]{\{#1\}}
\newcommand{\numprio}{d\xspace}

\title{Quasipolynomial Set-Based Symbolic Algorithms for~Parity~Games}

\author{Krishnendu Chatterjee\inst{1}
	\and Wolfgang Dvo\v{r}\'{a}k\inst{2}
	\and Monika Henzinger\inst{3}
	\and Alexander Svozil\inst{3}\thanks{A. S. is fully supported by the Vienna
	Science and Technology Fund (WWTF) through project ICT15-003. K.C. is supported by the Austrian Science
Fund (FWF) NFN Grant No S11407-N23 (RiSE/SHiNE) and an ERC Starting
grant (279307: Graph Games). For M.H the research leading to these
results has received funding from the European Research Council under
the European Union’s Seventh Framework Programme (FP/2007-2013) / ERC
Grant Agreement no. 340506.}
}

\institute{
IST Austria,
\email{krish.chat@ist.ac.at}
\and
Institute of Logic and Computation, TU Wien,
\email{dvorak@dbai.tuwien.ac.at}
\and Faculty of Computer Science, University of Vienna,
\email{\{monika.henzinger,alexander.svozil\}@univie.ac.at}
}

\authorrunning{Chatterjee, Dvo\v{r}\'{a}k, Henzinger and Svozil}
\titlerunning{Quasipolynomial Set-Based Symbolic Algorithms for~Parity~Games}

\begin{document}

\maketitle

\begin{abstract}
Solving parity games, which are equivalent to modal $\mu$-calculus model checking,
is a central algorithmic problem in formal methods, with applications in
reactive synthesis, program repair, verification of branching-time
properties, etc. Besides the standard computation model with the explicit
representation of games, another important theoretical model of
computation is that of \emph{set-based symbolic algorithms}. Set-based
symbolic algorithms use basic set operations and one-step predecessor
operations on the implicit description of games, rather than the explicit
representation. The significance of symbolic algorithms is that they
provide scalable algorithms for large finite-state systems, as well as for
infinite-state systems with finite quotient. Consider parity games on
graphs with $n$ vertices and parity conditions with $\numprio$ priorities.
While there is a rich literature of explicit algorithms for parity games,
the main results for set-based symbolic algorithms are as follows: (a)~the
basic algorithm that requires $O(n^\numprio)$ symbolic operations and
$O(\numprio)$ symbolic space; and (b)~an improved algorithm that requires
$O(n^{\numprio/3+1})$ symbolic operations and $O(n)$ symbolic space.  In
this work, our contributions are as follows: (1)~We present a black-box
set-based symbolic algorithm based on the explicit progress measure
algorithm. Two important consequences of our algorithm are as follows: (a)~a
set-based symbolic algorithm for parity games that requires
quasi-polynomially many symbolic operations and $O(n)$ symbolic space; and
(b)~any future improvement in progress measure based explicit algorithms
immediately imply an efficiency improvement in our set-based symbolic
algorithm for parity games. (2)~We present a set-based symbolic algorithm
that requires quasi-polynomially many symbolic operations and $O(\numprio
\cdot \log n)$ symbolic space. Moreover, for the important special case of
$\numprio \leq \log n$, our algorithm requires only polynomially many
symbolic operations and poly-logarithmic symbolic space.
\end{abstract}

\pagebreak
\section{Introduction}\label{sec:intro}

In this work, we present new contributions related to algorithms for parity
games in the set-based symbolic model of computation. 

\smallskip\noindent{\em Parity games.} Games on graphs are central in many
applications in computer science, especially, in the formal analysis of reactive
systems. The vertices of the graph represent states of the system, the edges
represent transitions of the system, the infinite paths of the graph represent
traces of the system and the players represent the interacting agents.  The
reactive synthesis problem (Church's problem~\cite{Church62}) is equivalent to
constructing a {\em winning strategy} in a graph
game~\cite{BuchiL69,RamadgeW87,PnueliR89}. Besides reactive synthesis, the
game graph problem has been used in many other applications, such as
(1)~verification of branching-time properties~\cite{EmersonJ91},
(2)~verification of open systems~\cite{AlurHK02}, 
(3)~simulation and refinement between reactive systems~\cite{Milner71,HenzingerKR02,AHKV98};
(4)~compatibility checking~\cite{InterfaceTheories}, 
(5)~program repair~\cite{JobstmannGB05},
(6)~synthesis of programs~\cite{CernyCHRS11}; to
name a few. Game graphs with {\em parity winning} conditions are particularly
important since all $\omega$-regular winning conditions (such as safety,
reachability, liveness, fairness) as well as all Linear-time Temporal Logic
(LTL) winning conditions can be translated into parity
conditions~\cite{Safra88,Safra89}.  In a parity game, every vertex of the game
graph is assigned a non-negative integer priority from
$\set{0,1,\ldots,\numprio-1}$, and a play is winning if the highest priority
visited infinitely often is even. Game graphs with parity conditions can model
all the applications mentioned above, and are also equivalent to the modal
$\mu$-calculus~\cite{Kozen83} model-checking problem~\cite{EmersonJ91}.  Thus
the parity games problem is a core algorithmic problem in formal methods, and
has received wide attention over the
decades~\cite{EmersonJ91,BrowneCJLM97,Seidl96,Jurdzinski00,VogeJ00,JurdzinskiPZ08,Schewe07,calude2017stoc,lazic2017qp}.

\smallskip\noindent{\em Models of computation: Explicit and symbolic
algorithms.} For the algorithmic analysis of parity games, two models of
computation are relevant.  First, the standard model of {\em explicit}
algorithms, where the algorithms operate on the explicit representation of the
game graph.  Second, the model of {\em implicit or symbolic} algorithms, where
the algorithms do not explicitly access the game graph but operate with a set
of predefined operations.  For parity games, the most relevant class of symbolic
algorithms are called {\em set-based symbolic algorithms}, where the allowed
symbolic operations are: (a)~basic set operations such as union, intersection,
complement, and inclusion; and (b)~one step predecessor (Pre) operations
(see~\cite{ClarkeGP99,deAlfaroHM01,HMR05}).

\smallskip\noindent{\em Significance of set-based symbolic algorithms.} We
describe the two most significant aspects of set-based symbolic algorithms.
\begin{compactenum}
\item Consider large scale finite-state systems, e.g., hardware circuits,
or programs with many Boolean variables or bounded-domain integer variables.
While the underlying game graph is described implicitly (such as program code),
the explicit game graph representation is huge (e.g., exponential in the number
of variables).  The implicit representation and symbolic algorithms often do not
incur the exponential blow-up, that is inevitable for algorithms that require
the explicit representation of the game graph.  Data-structures such as Binary
Decision Diagrams (BDDs)~\cite{bryant1986graph} (with well-established tools
e.g. CuDD~\cite{CuDD}) support symbolic algorithms that are used in verification
tools such as NuSMV~\cite{Cimatti2000}.

\item In several domains of formal analysis of infinite-state systems,
such as games of hybrid automata or timed automata, the underlying state space
is infinite, but there is a finite quotient. 
Symbolic algorithms provide a practical and scalable approach for the analysis of such systems:
For many applications the winning set is characterized by $\mu$-calculus formulas with one-step predecessor 
operations which immediately give the desired set-based symbolic algorithms~\cite{deAlfaroHM01,deAlfaroFHMS03}. 
Thus, the set-based symbolic model of computation is an equally important theoretical model of computation to be
studied as the explicit model.
\end{compactenum}

\smallskip\noindent{\em Symbolic resources.} In the explicit model of
computation, the two important resources are time and space. Similarly, in the
symbolic model of computation, the two important resources are the number of
symbolic operations and the symbolic space. 
\begin{compactitem}
\item {\em Symbolic operations:} Since a symbolic algorithm uses a set of
predefined operations, instead of time complexity, the first efficiency measure
for a symbolic algorithm is the number of symbolic operations required.  Note
that basic set operations (that only involve variables of the current state) are less resource
intensive compared to the predecessor operations (that involve both variables of
the current and of the next state).  Thus, in our analysis, we will distinguish
between the number of basic set operations and the number of predecessor
operations. 

\item {\em Symbolic space:} We refer to the number of sets stored by a set-based symbolic algorithm as the
symbolic space for the following reason: A set in the symbolic model is considered to be unit space, for
example, a set that contains all vertices, or a set that contain all vertices
where the first variable is true, represent each $\Theta(n)$ vertices, but can be
represented as BDD of constant size.  While the size of a set and its symbolic
representation is notoriously hard to characterize (e.g., for BDDs it can
depend on the variable reordering), in the theoretical model of computation every
set is considered of unit symbolic space, and the symbolic space requirement is
thus the maximum number of sets required by a symbolic algorithm. 
\end{compactitem}
The goal is to find algorithms that minimize the symbolic space (ideally
poly-logarithmic) and the symbolic operations.

\smallskip\noindent{\em Previous results.}
We summarize the main previous results for parity games on graphs 
with $n$ vertices, $m$ edges, and $\numprio$ priorities. 
To be concise in the following discussion, 
we ignore denominators in~$\numprio$ in the bounds.
\begin{compactitem}

\item \emph{Explicit algorithms.} The classical algorithm for parity games
requires $O(n^{\numprio-1}\cdot m)$ time and linear
space~\cite{Zielonka98,McNaughton93}, which was then improved by the
small progress measure algorithm that requires $O(n^{\numprio/2}\cdot m)$ time
and $O(\numprio \cdot n)$ space~\cite{Jurdzinski00}.  Many improvements
have been achieved since then, such as the big-step algorithm~\cite{Schewe07},
the sub-exponential time algorithm~\cite{JurdzinskiPZ08}, an improved algorithm
for dense graphs~\cite{ChatterjeeHL15}, and the strategy-improvement
algorithm~\cite{VogeJ00}, but the most important breakthrough was achieved last year
where a quasi-polynomial time $O(n^{\lceil \log\numprio \rceil+6})$ algorithm
was obtained~\cite{calude2017stoc}.  While the original algorithm
of~\cite{calude2017stoc} required quasi-polynomial time and space, a succinct
small progress measure based algorithm~\cite{lazic2017qp} and value-iteration
based approach~\cite{FearnleyJS0W17} achieve the quasi-polynomial time bound
with quasi-linear space. However, all of the above algorithms are inherently
explicit algorithms.  

\item \emph{Set-based symbolic algorithms.} The basic set-based symbolic
algorithm (based on the direct evaluation of the nested fixed point of the
$\mu$-calculus formula) requires $O(n^\numprio)$ symbolic
operations and $O(\numprio)$ space~\cite{EmersonL86}.  In a breakthrough
result~\cite{BrowneCJLM97} presented a set-based symbolic algorithm that
requires $O(n^{\numprio/2+1})$ symbolic operations and $O(n^{\numprio/2+1})$
symbolic space (for a simplified exposition see~\cite{Seidl96}). In recent
work~\cite{chatterjee2017symbolic}, a new set-based symbolic algorithm was
presented that requires $O(n^{\numprio/3+1})$ symbolic operations and $O(n)$
symbolic space, where the symbolic space requirement is $O(n)$ even with a constant
number of priorities. 

\end{compactitem}

\smallskip\noindent{\em Open questions.}
Despite the wealth of results for parity games, many fundamental algorithmic
questions are still open. Besides the major and long-standing open question of
the existence of a polynomial-time algorithm for parity games, two important open
questions in relation to set-based symbolic algorithms are as follows:
\begin{compactitem} 
\item {\em Question~1.} Does there exist a set-based symbolic algorithm that
requires only quasi-polynomially many symbolic operations?

\item {\em Question~2.} Given the $O(\numprio)$ symbolic space requirement of
the basic algorithm, whereas all other algorithms require at least $O(n)$
space (even for a constant number of priorities) an important question is:
Does there exist a set-based symbolic algorithm that requires
$\widetilde{O}(\numprio)$ symbolic space (note that $\widetilde{O}$ hides poly-logarithmic factors), but beats the number of symbolic
operations of the basic algorithm?
This question is especially relevant since in
many applications the number of priorities is small, e.g., in
determinization of $\omega$-automata, the number of priorities is
logarithmic in the size of the automata~\cite{Safra88}.
\end{compactitem}

\begin{table}[!t]
	\centering
	\renewcommand{\arraystretch}{1.3}
	\caption{Set-Based Symbolic Algorithms for Parity
		Games.}\label{tab:comparison}
	\begin{tabular}{@{}lcc@{}}
		\toprule
		reference & symbolic operations & symbolic space\\
		\midrule
		\cite{EmersonL86,Zielonka98} & $O(n^\numprio)$ & $O(\numprio)$ \\
		\cite{BrowneCJLM97,Seidl96} & $O(n^{\numprio/2+1})$ &
		$O(n^{\numprio/2+1})$\\
		\cite{chatterjee2017symbolic} & $O(n^{\numprio/3+1})$ & $O(n)$ \\
		Thm.~\ref{thm:ordered_main},\! \ref{thm:reducedspace_main} & $n^{O(\log
		\numprio)}$ & $O(\numprio \log n)$\\
		\bottomrule
	\end{tabular}
\end{table}

\smallskip\noindent{\em Our contributions.} 
In this work, we not only answer the above open questions (Question~1 and
Question~2) in the affirmative but also show that both can be achieved by the same
algorithm:
\begin{compactitem}
\item First, we present a black-box set-based symbolic algorithm based on
explicit progress measure algorithm for parity games that use $O(n)$
symbolic space and $O(n^{O(\log \numprio)})$ symbolic operations.  
There are two important consequences of our algorithm: (a)~First, given the
ordered progress measure algorithm (which is an explicit algorithm), as a consequence of our
black-box algorithm, we obtain a set-based symbolic algorithm for parity
games that requires quasi-polynomially many symbolic operations and $O(n)$
symbolic space.  (b)~Second, any future improvement in progress measure
based explicit algorithm (such as polynomial-time progress measure
algorithm) would immediately imply the same improvement for set-based
symbolic algorithms.  Thus we answer Question~1 in affirmative and also
show that improvements in explicit progress measure algorithms carry over to
symbolic algorithms.

\item Second, we present a set-based symbolic algorithm that requires quasi-polynomially
many symbolic operations and $O(\numprio \cdot \log
n)=\widetilde{O}(\numprio)$ symbolic space.  Thus we not only answer
Question~2 in affirmative, we also match the number of
symbolic operations with the current best-known bounds for explicit
algorithms.  Moreover, for the important case of $\numprio \leq \log
n$, our algorithm requires polynomially many symbolic operations and
poly-logarithmic symbolic space.
\end{compactitem}
We compare our main results with previous set-based symbolic algorithms in Table~\ref{tab:comparison}.

\smallskip\noindent{\em Symbolic Implementations.}
Recently, symbolic algorithms for parity games received attention from a practical
perspective: First, three explicit algorithms (Zielonka's recursive algorithm,
Priority Promotion~\cite{BenerecettiDM16} and Fixpoint-Iteration~\cite{BFL14}) were converted to symbolic implementations~\cite{SWW18}. 
The symbolic solvers had a huge performance gain compared to the corresponding
explicit solvers on a number of practical instances. 
Second, four symbolic algorithms to solve parity games were compared to their
explicit versions (Zielonka's recursive algorithm, small progress measure and an automata-based algorithm~\cite{KV98,SMPV16})
\cite{SMV18}. For the symbolic versions of the small progress measure, two implementations were considered: (i) Symbolic
Small Progress Measure using Algebraic Decision Diagrams~\cite{BKV04} and (ii) the Set-Based Symbolic
Small Progress Measure~\cite{chatterjee2017symbolic}. The symbolic
algorithms were shown to perform better in several structured instances.

\smallskip\noindent{\em Other related works.} 
Besides the discussed theoretical results on parity games,
there are several practical approaches for parity games, such as, 
(a)~accelerated progress measure~\cite{deAlfaroF07}, 
(b)~quasi-dominion~\cite{BenerecettiDM16}, 
(c)~identifying winning cores~\cite{Vester16},
(d)~BDD-based approaches~\cite{KP12,KP14}, and 
(e)~an extensive comparison of various solvers~\cite{dijk2018oink}.
A straightforward symbolic implementation (not set-based) of small progress measure was
done in~\cite{BKV04} using \emph{Algebraic Decision Diagrams (ADDs)} and BDDs. 
Unfortunately, the running time is not comparable with our results
as using ADDs breaks the boundaries of the Set-Based Symbolic Model: ADDs
can be seen as BDDs which allow the use of a finite domain at the leaves~\cite{BFGHMPS1997}.
Recently, a novel approach for solving parity games in quasi-polynomial time
which uses the register-index was introduced~\cite{Lehtinen18}.
Moreover,~\cite{Lehtinen18} presents a $\mu$-calculus formula describing the
winning regions of the parity game with alternation depth based on the
register-index. The existence of such a $\mu$-calculus formula does not immediately imply
a quasi-polynomial set-based symbolic algorithm due to constructing the formula
using the register-index. 

\smallskip\noindent
Our work considers the theoretical model of symbolic computation 
and presents a black-box algorithm as well as a quasi-polynomial algorithm,
matching the best-known bounds of explicit algorithms. 
Thus our work makes a significant contribution towards the theoretical understanding of symbolic computation for parity games.

\section{Preliminaries}

We follow a similar notation as~\cite{chatterjee2017symbolic}.

\subsection{Basic Definitions}

\noindent\emph{Game Graphs.}
A game graph is a graph $\GG = (V,E, \langle V_\E, V_\O \rangle)$ where the 
vertices $V$ are partitioned into player-\E vertices $V_\E$ and 
player-\O vertices $V_\O$, i.e., $V = V_\E \cup V_\O$. 
Let $\Out(v)$ describe the set of successor vertices of $v$. The set
$\In(v)$ describes the set of predecessors of the vertex $v$.
More formally $\Out(v) = \{ w \in V \mid (v,w) \in E \}$ and $\In(v) = \{w \in V 
	\mid (w,v) \in E \}$. We assume without loss of generality that every vertex has
an outgoing edge. We denote the number of vertices with $n$ and the number of
edges with $m$.

\smallskip\noindent\emph{Plays.}
Let $\GG = (V,E, \langle V_\E, V_\O \rangle)$ be a game graph.
Initially, a token is placed on a vertex $v_0 \in V$. When $v \in V_z$ for $z
\in \{ \E, \O \}$, player $z$ moves the token along one of the edges to a vertex in $\Out(v)$. Formally, a
\emph{play} is an infinite sequence $\langle v_0,v_1,v_2,v_3,\dots,\rangle$ where for every $i\geq 0$ the
following condition holds: $(v_i,v_{i+1}) \in E$. 

\smallskip\noindent\emph{Parity Game.}
A \emph{parity game} $\PG$ with $\numprio$ priorities is a game graph $\Gamma$ with a
function $\alpha$ that assigns each vertex a priority, i.e., $\alpha: V \mapsto
\{0,1,2,\dots \numprio-1\}$ where $\numprio \in \N$ and $\numprio > 0$ and $\PG = (\Gamma, \alpha)$.
The set $C$ is the set of all priorities. Let $\rho$ be a play of $\PG$. Player \E
\emph{wins} $\rho$ if the highest priority occurring
infinitely often is even. Player \O, on the other hand, wins $\rho$ if
the highest priority occurring infinitely often is odd.
Let $V_i$ for $0 \leq i \leq \numprio-1$ denote the vertices in $\PG$ with priority $i$. 
Formally, we define $V_i = \{ v \in V \mid \alpha(v) = i \}$.

\smallskip\noindent\emph{Strategies.}
A \emph{strategy} for player $z \in \{\E,\O\}$ is a function that extends a finite
prefix of a play which ends at vertex $v \in V_z$ by appending a vertex $v' \in
\Out(v)$. A \emph{memoryless strategy} is a strategy that depends only on the
last vertex of a play. This corresponds to a function $\sigma_z: V_z \mapsto V$
such that $\sigma_z(v) \in \Out(v)$ for all $v \in V_z$. 
The results from~\cite{EmersonJ91,McNaughton93} show that it is sufficient to consider memoryless strategies
for parity games. We shall therefore from now on only consider memoryless
strategies. A starting vertex $s \in V$, a player-\E strategy $\sigma$, and a player-\O strategy $\pi$
describe a unique play $\omega(s,\sigma,\pi) = \langle v_0, v_1, v_2 \dots
\rangle$ in a game graph. It is defined as follows: $v_0 = s$ and for all $i\geq
0$, if $v_i \in V_\E$ then $v_{i+1} = \sigma(v_i)$ and if $v_i \in V_\O$ then
$v_{i+1} = \pi(v_i)$. 

\smallskip\noindent\emph{Winning Strategies and Winning Sets.}
A strategy $\sigma$ is \emph{winning} at a start vertex $s \in V$ 
for player \E iff for all strategies $\pi$ of player \O, player \E wins
the play $\omega(s,\sigma,\pi)$. If there exists a winning strategy at a start vertex $s \in V$ 
for player $z \in \{\O,\E\}$, $s$ is part of the winning set of player $z$, $\win_z$.
Every vertex is winning for exactly one of the players~\cite{EmersonJ91,MOS91}.
In this work, we study the problem of computing the winning sets for the two
players.

\subsection{Symbolic Model of Computation}

In the set-based symbolic model, the game graph is not accessed explicitly but with
set-based symbolic operations. The resources in the
symbolic model of computation are characterized by the number of set-based
symbolic operations and the set-based space. 

\smallskip\noindent\emph{Set-Based Symbolic Operations.}
A set-based symbolic algorithm is allowed to use the same mathematical, logical and memory
access operations as a regular RAM algorithm, except for the access to the input
graph.
Given an input game graph $\GG = (V, E, \langle V_\E, V_\O \rangle )$ 
and a set of vertices $S \subseteq V$, the game graph
$G$ can be accessed only by the following two types of operations:
\begin{enumerate}
	\item The \emph{basic set operation}: $\cup,\cap,\subseteq,\setminus$ and
		$=$.
	\item The \emph{one-step operation} to obtain the predecessors of
		the vertices of $S$ in $G$. In particular, we define the predecessor operation
		\begin{equation*}
			\Pre(S) = \{ v \in V \mid \Out(v) \cap S \neq \emptyset \}.
		\end{equation*}
		
\end{enumerate}
Let $z \in \{\E,\O\}$, then $\bar{z}=\O$ if $z=\E$ and $\bar{z}=\E$ if $z=\O$.
The \emph{controllable} predecessor operation for $z \in \{\E,\O\}$ is defined
as
\begin{equation*}
\CP_z(S) = \{ v \in V_z \mid \Out(v) \cap S \neq \emptyset \} \cup \{ v \in
	V_{\bar{z}} \mid \Out(v) \subseteq S \}.
\end{equation*}
The set $\CP_z(S)$ can be expressed using only $\Pre$ and
basic set operations.
Note that basic set operations (that only involve variables of the current state) are much cheaper
compared to the one-step operations (that involve both variables of
the current and of the next state).  Thus, in our analysis, we will distinguish
between the number of basic set operations and the number of one-step operations. 
Notice that one can define a one-step \emph{successor} operation (denoted $\mathit{Post}(\cdot)$) as well~\cite{ChatterjeeDHL18},
but for the algorithms presented in this work the given predecessor operation suffices.

\smallskip\noindent\emph{Set-based Symbolic Space.}
The basic unit of space for a set-based symbolic algorithm for game graphs are
sets~\cite{BrowneCJLM97,chatterjee2017symbolic}.  For example, a set can
be represented symbolically as one
BDD~\cite{bryant1986graph,bry92,BurchCMDH90,ClarkeMCH96,Somenzi99,ClarkeGP99,ClarkeGJLV03,GentiliniPP08,ChatterjeeHJS13}
and each such set is considered as unit space.  Consider for example a game
graph whose state-space consists of valuations of $N$-Boolean variables. The
set of all vertices is simply represented as a true BDD. Similarly, the set of all
vertices where the $k$th bit is false is represented by a BDD which depending
on the value of the $k$th bit chooses true or false. Again, this set can be
represented as a constant size BDD. 
Thus, even large sets can sometimes be represented as constant-size BDDs.
In general, the size of the smallest BDD representing a set is notoriously hard to
determine and depends on the variable reordering~\cite{ClarkeGP99}.
To obtain a clean theoretical model for the algorithmic analysis,
each set is represented as a unit data structure and
requires unit space.  Thus, for the \emph{space requirements} of a symbolic
algorithm, we count the maximal number of sets the algorithm stores
simultaneously and denote it as the {\em symbolic space}.

\subsection{The Progress Measure Algorithm}\label{ss:pm}

\emph{High-level intuition.} Let $\PG = (V,E, \langle V_\E,
V_\O \rangle, \alpha)$ be a parity game and let $(\W, \prec)$ be a finite total order
with a maximal element $\top$ and a minimal element $\min$. A \emph{ranking function} is a function
$f$ which maps every vertex in $V$ to a value in $\W$. The value of a
vertex $v$ with respect to the ranking function $f$ is called \emph{rank} of
$v$. The rank $f(v)$ of a vertex $v$ determines how ``close'' the vertex is to
being in $W_z$, the winning set of a fixed player $z$.
Initially, the rank of every vertex is the minimal value of $\W$. The
\emph{progress measure algorithm} iteratively
increases the rank of a vertex $v$ with an operator called $\Lift$ with
respect to the successors of $v$ and another function called $\lift$. The algorithm
terminates when the rank of no vertex can be increased any further, i.e., the
least fixed point of $\Lift$ is reached. We call the least simultaneous fixed
point of all $\Lift$-operators \emph{progress measure}. 
When the rank of a vertex is the maximal element of the total order it is
declared winning for player $z$. The rest of the vertices are declared winning for the adversarial player
$\z$. 

\smallskip\noindent\emph{Ranking Function.}
Let $\W$ be a total order with a minimal element $\mina$ and maximal element
$\top$. A ranking function is a function ${f: V \mapsto \W}$.

\smallskip\noindent\emph{The $\best$ function.} The \emph{$\best$} function
represents the ability of player $z$, given the token is at vertex $v$, 
to choose the vertex in $\Out(v)$ with the
maximal ranking function. Analogously, it constitutes the ability of
player $\z$, given the token is at vertex $v$, to choose the vertex in $\Out(v)$ with the
minimal ranking function. Formally, the function $\best$ is defined for a vertex $v$ and a
ranking function $f$ as follows:
\begin{equation*}
\best(f,v)  = \begin{cases}
	\min \{ f(w) \mid w \in \Out (v) \} & \text{if } v \in V_{\z}\\
	\max \{ f(w) \mid w \in \Out(v) \} & \text{if } v \in V_z
\end{cases}
\end{equation*}

\smallskip\noindent\emph{The $\lift$-function.}
The function $\lift: \W \times C \mapsto \W$ defines how the rank of
a vertex $v$ is increased according to the rank $r$ of a successor vertex,
and the priority $\alpha(v)$ of $v$. The $\lift$ function 
needs to be monotonic in the first argument. Notice that we do not need
information about the graph to compute the $\lift$ function. In all known
progress measures, the $\lift$ function is computable in constant time.

\smallskip\noindent\emph{The $\Lift$-operation.}
The $\Lift$-operation potentially increases the rank of a vertex $v$ according to
its priority $\alpha(v)$ and the rank of all its successors in the graph.\footnote{Notice that in the original definition~\cite{Jurdzinski00} the $\lift$ function is applied to all successors and the best of them is chosen subsequently.
As the $\lift$ is monotone in the first argument the two definitions are equivalent.}
\begin{equation*}
\Lift(f,v)(u) = \begin{cases}
	\lift(\best(f,v), \alpha(v)) & \text{if $u = v$}\\
	f(u)					  & \text{otherwise}
\end{cases}
\end{equation*}
A ranking function is
a \emph{progress measure} if it is the least simultaneous fixed point of all
$\Lift(\cdot,v)$-operators.

\smallskip\noindent\emph{The Progress Measure Algorithm.}
The progress measure algorithm initializes the ranking function $f$ with the
minimum element of $\W$. Then, the $\Lift(\cdot,v)$-operator is computed in an
arbitrary order regarding the vertices. The winning set of player z can be obtained from a progress
measure by selecting those vertices whose rank is $\top$.
Notice that we need to define the total order $(\W,\prec)$ and a function
$\lift$ to initialize the algorithm. 

For example, the following instantiations of the progress measure algorithm
determine the winning set of a parity game:
(i) Small Progress Measure~\cite{Jurdzinski00}, (ii) Succinct Progress
Measure~\cite{lazic2017qp} and the (iii) The Ordered
Approach~\cite{FearnleyJS0W17}. 
The running time is dominated by the size of $\W$.
For a discussion on the state-of-the-art size of $\W$ we refer
the reader to Remark~\ref{rem:bound}. 

\section{Set-Based Symbolic Black Box Progress Measure Algorithm}\label{sec:bb3}

In this section, we briefly present a basic version of a set-based symbolic progress measure
algorithm. The key idea is to compute the Lift-operation with a
set-based symbolic algorithm. Then, we improve the basic 
version to obtain our black box set-based symbolic progress measure algorithm.
Finally, we prove its correctness and analyze the symbolic resources.

\smallskip\noindent\emph{A basic Black Box Algorithm.}
Throughout the algorithm, we maintain the family $\mathbf{S}$ of sets of vertices, which contains a set 
for every element in $\W$, i.e.,  $\S = \{ S_r \mid r \in \W\}$. 
Intuitively, a vertex $v \in S_r$ has rank $f(v) = r$.
Initially, we put each vertex into the set with the minimal value of $\W$,
i.e., $S_\mina$. In each iteration, we consider all non-empty sets $S_r \in \S$: 
The algorithm checks if the ranking function of the predecessors of the vertices
in $S_r$ must be increased, i.e. $\Lift(f,v)(v) \succ f(v)$ where $v \in \Pre(S_r)$, 
and if so, performs the $\Lift$-operator for the predecessors.
We repeat this step until the algorithm arrives at a fixed point.

\smallskip\noindent\emph{Performing a Lift operation.} 
To compute the $\Lift$-operation in the set-based symbolic setting we need to
compute two functions for the predecessors of $S_r$: (1)~the $\lift$-function and (2)~the $\best$-function.
By definition, the $\lift$-function does not access the game graph or vertices thereof.
Thus we can compute the $\lift$-function without the use of symbolic operations. 
To compute the $\best$-function we need access to the game graph.
In turns out it is simpler to compute the vertices with $\best(f,v) \succeq r$
rather than the vertices with $\best(f,v) = r$. 
Thus, we lift all vertices $v$ with $\best(f,v) \succeq r$ to the rank $\lift(r,\alpha(v))$. 
To this end, we first compute the set $S_{\succeq r} = \bigcup_{l \succeq r} S_l$ of vertices with rank $\succeq r$.
Then, we compute $P = \CP_z(S_{\succeq r})$ and, hence, the set $P$ comprises the vertices $v \in P$ with
$\best(f,v) \succeq r$. Finally, to compute $\Lift(f,v)$, for each $c \in
C$, we consider the vertices of $P$ with priority $c$, i.e., the set  $(P \cap V_c)$, and add  them to the set $S_{\lift(r,c)}$. 
Notice that we lift each vertex $v$ to $\lift(r,\alpha(v))$
where $r = \best(f,v)$ as we consider all non-empty sets $S_r \in \S$. No vertex
$v$ will be lifted to a set higher than $\lift(r,\alpha(v))$ where  $r = \best(f,v)$ due to the monotonicity
of the lift function in the first argument.
If after an iteration of the algorithm a vertex appears in several sets of $\mathbf{S}$ 
we only keep it in the set corresponding to the largest rank and remove it from all the other sets.

\subsection{Improving the Basic Algorithm}\label{ss:algorithm}

In this section, we improve the basic Algorithm by (a) reducing the symbolic
space from $O(|\mathcal{W}|)$ to $O(n)$ 
and (b) by reducing the number of symbolic operations required to compute the fixed point.

\smallskip\noindent\emph{Key Idea.}
The naive algorithm considers each non-empty set $S_r$ in iteration
$i\+1$ again no matter if $S_r$ has been changed in iteration $i$ or not.
Notice that we only need to consider the predecessors of the set $S_r$ again when
the set $S_{\succeq r}$ in iteration $i\+1$ contains 
additional vertices compared to the set $S_{\succeq r}$ in iteration $i$.
To overcome this weakness, we propose Algorithm~\ref{alg:blackbox3}. In this
algorithm, we introduce a data structure called $D$. In the data structure $D$ we keep track
of the sets $S_{\succeq r}$ instead of the sets $S_{r}$. The set $S_{\succeq r}$
contains all vertices with a rank greater or equal than $r$. Furthermore, we separately keep track of the elements $r \in \W$ where the set
$S_{\succeq r}$ changed since the last time $S_{\succeq r}$ was selected to be processed. 
These elements of $\W$ are called \emph{active}.
Moreover, if we have two sets $S_{\succeq r} = S_{\succeq r'}$ with $r \prec r'$
there is no need to process the set $S_{\succeq r}$
because $\lift(r',c) \succeq \lift(r,c)$ holds due to the monotonicity of the
$\lift$ function. To summarize, we precisely store a set $S_{\succeq r}$ if there is no $r'$
with $r \prec r'$ and $S_{\succeq r} = S_{\succeq r'}$.
Notice, that this instantly gives us a bound on the symbolic space of
$O(n)$. 

\smallskip\noindent\emph{Algorithm Description.} 
In Algorithm~\ref{alg:blackbox3} we use the data structure $D$ to manage the active $r \in \W$ and in each iteration of 
the outer while-loop we process the corresponding set $S_{\succeq r}$ of such an
$r \in \W$. 
We first compute $P = \CP_z(S_{\succeq r})$, then,
for each $c \in C$  we compute $r' = \lift(r,c)$
and update the set $S_{\succeq r'}$ by adding $P \cap V_c$. 
The inner while-loop ensures that $P \cap V_c$ is also added to all the sets $S_{\succeq r'}$ with $r' \prec r$ and 
the properties of the data structure are maintained, i.e.,
(a) all active elements are in the active list, and
(b) exactly those $r \in \W$ with $S_{\succeq r} \supset S_{\succeq r'}$, for $r
\prec r'$ are stored in $D$.

\begin{algorithm}[t]
	\SetKwInOut{Input}{input}
	\SetKwInOut{Output}{output}
	\SetAlgoVlined
	\Input{Parity Game $\PG$}
	\caption{Black Box Set-Based Symbolic Progress Measure}\label{alg:blackbox3}
	Initialize data structure $D$\; \label{alg:bb3:dinit}
	D.activate$(min)$\;\label{alg:bb3:activeinit}
	\While{$r \gets D.popActiveSet()$\label{alg:bb3:while1}}{\label{alg:bb3:deactivate}
		$S_{\succeq r}  \gets D.getSet(r)$\;
		$P \gets \CP_z(S_{\succeq r})$\;\label{alg:bb3:cp}
		\For{$c \in C$} {\label{alg:bb3:for1}
			$r' \gets \lift(r,c)$\;\label{alg:bb3:lift}
			$S_{\succeq r'} \gets D.getSet(r')$\;
			\While{$P \cap V_c \not\subseteq S_{\succeq r'}\label{alg:bb3:while2}$}
			{
				$S_{\succeq r'} \gets S_{\succeq r'} \cup (P \cap V_c)$\label{alg:bb3:liftvertices}\;
				$S_{\succeq next(r')} \gets D.getSet(D.getNext(r'))$\;
				\If{$r' = \top$ or $S_{\succeq r'} \supset S_{\succeq next(r')}$\label{alg:bb3:if1}}
				{
					\tcp{$S_{\succeq r'}$ is a super set of $S_{\succeq
							next(r')}$ and thus we save it}
					$D.update(r', S_{\succeq r'})$;\label{alg:bb3:store}
					$D.activate(r')$\;\label{alg:bb3:activate}
				}				
				$S_{\succeq prev(r')} \gets D.getSet(D.getPrevious(r'))$\;
				\If{$S_{\succeq r'} = S_{\succeq prev(r')} \cup (P \cap V_c)$}{\label{alg:bb3:else1}
					\tcp{We only keep sets which are different}
					$D.removeSet(D.getPrevious(r'))$\;\label{alg:bb3:removed}
				}
				$r'\gets D.getPrevious(r')$\;\label{alg:bb3:while2end}
				$S_{\succeq r'}  \gets
				D.getSet(r')\;$\;\label{alg:bb3:while2end2}
			}    
		}
	}
	\Return $S_\top$
\end{algorithm}

\smallskip\noindent\emph{Active Elements.} Intuitively, an element $r \in \W$ is \emph{active}
if $S_{\succeq r} \supset S_{\succeq r'}$, for all $r'$ where $r \prec r'$ 
and $S_{\succeq r}$ has been changed since the last time 
$S_{\succeq r}$ was selected at Line~\ref{alg:bb3:while1} of Algorithm~\ref{alg:blackbox3}.
We define active elements more formally later.

\begin{datastructure}\label{alg:bb3:ds}
Our algorithm relies on a data structure $D$, which supports the following operations:
\begin{compactitem}
\item $D.popActiveSet()$ returns an element $r \in \W$ marked as active and makes it
	inactive. If all elements are inactive, returns false.

      \item $D.getSet(r)$ returns the set $S_{\succeq r}$. 		  
      
      \item $D.getNext(r)$ returns the smallest $r'$ with $S_{\succeq r} \supset  S_{\succeq r'}$.
      
      \item $D.getPrevious(r)$ returns the largest $r'$ where $S_{\succeq r} \subset S_{\succeq r'}$.
    
      \item $D.removeSet(r)$ 
			     marks $r$ as inactive. 

      \item $D.activate(r)$ marks $r$ as active.
      
      \item $D.update(r,S)$  updates the set $S_{\succeq r}$ to $S$, i.e., $D.getSet(r)$ returns $S$.
							 Moreover, all sets $S_{\succeq r'}$ with $r' \prec r$ and $S_{\succeq r'}=S_{\succeq r}$ beforehand
							 are updated to $S$ as well.
  \end{compactitem}
 We initialize $D$ with $D.update(\mina,V)$ and $D.update(\top, \emptyset)$.
\end{datastructure}

We can define \emph{active} elements formally now as the definition depends on
$D$.
\begin{definition}
	Let $S^0_{\succeq r}=D.getSet(r)$ be the set stored in $D$ for  $S_{\succeq r}$ 
	after the initialization of $D$, 
	and let $S^i_{\succeq r}=D.getSet(r)$ be the set
	stored in $D$ for $S_{\succeq r}$ after the $i$-th iteration of the
	while-loop at Line~\ref{alg:bb3:while1}. 
         An element $r \in \W$ is \emph{active} after the $i$-th iteration
	of the while-loop 
	if 
	(i) for all $r' \in \W$ where $r \prec r'$ we have 
	 $S^i_{\succeq r} \supset S^i_{\succeq r'}$ and
	(ii) there is a $j < i$ such that $S^i_{\succeq r} \supset S^j_{\succeq r}$ 
	and for all $j < j' \leq i$ the set $S_{\succeq r}$ 
	is not selected in Line~\ref{alg:bb3:while1} in the $j'$-th iteration.
	Additionally, we consider $r=\min$ as active before the first iteration. An
	element $r \in \W$ is \emph{inactive} if it is not active. 
\end{definition}

Notice that, in Algorithm~\ref{alg:blackbox3} an $r \in \W$ is active iff $r$ is marked as active in $D$. 
The algorithm ensures this in a very direct way.  
At the beginning only $\min \in W$ is active, which is also marked as active in $D$
in the initial phase of the algorithm.
Whenever some vertices are added to a set $S_{\succeq r}$, it
is tested whether $S_{\succeq r}$ is larger than its successor and if so $r$
is activated (Lines~\ref{alg:bb3:if1}-\ref{alg:bb3:activate}).  
On the other hand, if something is added to the successor of
$S_{\succeq r}$ in the data structure $D$ then the algorithm tests whether the two sets are equal and
if so $r$ is rendered inactive (Lines~\ref{alg:bb3:else1}-\ref{alg:bb3:removed}). 

\smallskip\noindent\emph{Implementation of the data structure $D$.}
The data structure uses an AVL-tree and 
a doubly linked list called ``active list'' that keeps track of the active elements.
The nodes of the tree contain a pointer to the corresponding set $S_{\succeq r}$
and to the corresponding element in the active list.

\begin{compactitem}
	\item Initialization of the data structure $D$: Create the AVL tree with the elements $min$ and $\top$.
	      The former points to the set of all vertices and the latter to the empty set.
	      Create the doubly linked list called ``active list'' as an empty list.
	      
	\item $D.popActiveSet()$: Return the first element from the active list and
		remove it from the active list. If the list is empty, return false.
		
	\item $D.getSet(r)$:  Searches the AVL tree for $r$ or for the next greater
		element (w.r.t. $\succeq$). Then we return the set by using the pointer
		we stored at the node. 
	\item $D.getNext(r)$: 
	        First performs $D.getSet(r)$ and then computes the inorder successor in the 
	        AVL-tree. This corresponds to the next greater node w.r.t. $\succeq$.
	        	
	\item $D.getPrevious(r)$: 
	        First performs $D.getSet(r)$ and then computes the inorder predecessor in the 
	        AVL-tree. This corresponds to the next smaller node w.r.t. $\succeq$.
		
	\item $D.removeSet(r)$: This operation needs the element $r$ to be stored in the
		AVL tree. Search the AVL tree for $r$. Remove the corresponding element from
		the active list and the AVL Tree. 
		
	\item $D.activate(r)$: This operation needs the element $r$ to be stored in the
		AVL tree.
		Add $r$ to the active list and add pointers to the AVL-tree.
		The element in the active list contains a pointer to the tree element and vice versa.
	
	\item $D.update(r,S)$ :        
	        Perform $S_{\succeq r} \gets D.getSet(r)$: 
	        If $r$ is contained in the AVL tree then update $S_{\succeq r}$ to $S$.
	        Otherwise, insert $r$ as a new element and let the element point to $S$.		
\end{compactitem}

We initialize the data structure $D$ with $min \in \W$ and
$\top \in \W$. Thus, whenever we query $D$ for a value $r \in \W$ we find
it or there exists an $r' \succ r$ which is in $D$.

\smallskip\noindent\emph{Analysis of the data structure $D$.}
The data structure can be implemented with an AVL-tree and 
a doubly linked list called ``active list'' that keeps track of the active elements such that
all of the operations can be performed in $O(\log n)$:
when the algorithm computes $D.update(r,S)$ we store $r$ and a pointer
to the set $S$ as a node in the AVL tree.
By construction, the algorithm only stores pointers to different sets and when we additionally preserve anti-monotonicity among the sets 
we only store $\leq n$ sets. Therefore, the AVL tree has only $\leq n$
nodes with pointers to the corresponding sets and searching for a set with
the operation $D.getSet(r)$ only adds a factor of $\log n$ to the non-symbolic operations
when we store $r$ as key with a pointer to $S_{\succeq r}$ in the AVL-tree.
Moreover, we maintain pointers between the elements of the active list and the corresponding vertices in the AVL tree.

\begin{remarkb}\label{remark1}
	The described algorithm is based on a data structure $D$ 
	which keeps track of the sets that will be processed at some point later in time. 
	Note that this data structure
	does not access the game graph but only stores pointers to sets that the
	Algorithm~\ref{alg:blackbox3} maintains. 
	The size of the AVL tree implementing $D$ is proportional to the symbolic space of the algorithm.
\end{remarkb}

\subsubsection{Correctness}\label{ss:correctness}
In order to prove the correctness of Algorithm~\ref{alg:blackbox3} we tacitly assume that the algorithm terminates.
An upper bound on the running time is then shown in Proposition~\ref{prop:numberactivations}.

\begin{proposition}[Correctness.]\label{lem:bb3:correctness}
	Let $\PG$ be a parity game.
	Given a finite total order $(\W,\prec)$ with minimum element $\mina$, a	maximum element $\top$ 
	and a monotonic function $\lift:\W \times C \mapsto \W$ 
	Algorithm~\ref{alg:blackbox3} computes the least simultaneous fixed point of
	all $\Lift(\cdot,v)$-operators.	
\end{proposition}

To prove the correctness of Algorithm~\ref{alg:blackbox3}, we prove that when
Algorithm~\ref{alg:blackbox3} terminates, the function $\rho(v) =
\max\{r \in \W \mid v \in S_{\succeq r}\}$ is equal to the least simultaneous fixed point of
all $\Lift(\cdot,v)$-operators. We show that when the properties described in
Invariant~\ref{inv:correctness} hold, the function $\rho$ is equal to the 
least fixed point at the termination of the algorithm. Then, we prove that we maintain the properties of Invariant~\ref{inv:correctness}.

\begin{invariant}\label{inv:correctness}
	Let $\trho$ be the least simultaneous fixed point of
	$\Lift(\cdot,v)$ and $\rho(v) = \max\{r \in \W \mid v \in S_{\succeq r}\}$
	be the ranking function w.r.t.\ the sets $S_{\succeq r}$ that are
	maintained by the algorithm.
	\begin{compactenum}
		\item Before each iteration of the while-loop at
			Line~\ref{alg:bb3:while1} we have $S_{\succeq r_2} \subseteq
			S_{\succeq r_1}$ for all $r_1 \preceq r_2$ (anti-monotonicity).
		\item Throughout Algorithm~\ref{alg:blackbox3} we have $\trho(v) \succeq
			\rho(v)$ for all $v \in V$.
		\item For all $r \in \W$: (a) $r$ is active or (b) for all $v \in \CP_z(S_{\succeq r})\!:\\ 
					\text{ if } \best(\rho,v) = r$, then	$\rho(v) = \Lift(\rho,v)(v)$.
	\end{compactenum}
\end{invariant}

In the following paragraph, we describe the intuition of
Invariant~\ref{inv:correctness}. Then,
we show that the properties of Invariant~\ref{inv:correctness} are sufficient to obtain the correctness of
Algorithm~\ref{alg:blackbox3}. Finally, we prove that each property holds
during the while-loop at Line~\ref{alg:bb3:while1}.

\smallskip\noindent\emph{Intuitive Description.} The intuitive description is as follows:
\begin{compactenum}
	\item Ensures that the sets $S_{\succeq r}$ contain the correct elements. Having the
		sets $S_{\succeq r}$ allows computing $\best(f,v) \succeq r$ as discussed at the
		beginning of the section.
	\item Guarantees that $\rho$ is a lower
		bound on $\tilde{\rho}$ throughout the algorithm.
	\item When an $r \in \W$ is not active, 
		the rank of no vertex can be increased by applying $\lift$ to the
		vertices which have $\best(\rho,v) = r$.
\end{compactenum}
When the algorithm terminates, all $r \in \W$ are inactive and $\rho$ is 
a fixed point of all $\Lift(\rho,v)$ by condition (3b). The next lemma proves
that Algorithm~\ref{alg:blackbox3} computes the least simultaneous fixed point of
all $\Lift(\cdot, v)$ operators for a parity game.

\begin{lemma}[The Invariant is sufficient]
	Let the $\lift$ function be monotonic in the first argument and $(\W, \prec)$ be a total order.
	The ranking function $\rho$ at termination of Algorithm~\ref{alg:blackbox3} is
	equal to the least simultaneous fixed point of all $\Lift(\cdot,v)$-operators
	for the given parity game $\PG$.
\end{lemma}
\begin{proof}
	Consider the ranking function $\rho(v) = \max\{r \in \W \mid v \in S_{\succeq r}\}$
	computed by Algorithm~\ref{alg:blackbox3}.
	By Invariant~\ref{inv:correctness}(2) we have $\trho(v) \succeq \rho(v)$ for all $v \in V$. 
	We next show that $\rho(v)$ is a fixed point of $\Lift(\rho,v)$ for all $v \in V$.
	When the algorithm terminates, no $r \in \W$ is active. 
	Consider an arbitrary $v$ and let $r=\best(\rho,v)$.
	Now, as the set $r$ is not active, by Invariant~\ref{inv:correctness}(3b),
	we have $\rho(v) = \Lift(\rho,v)(v)$.
	Thus $\rho(v)$ is a fixed point of $\Lift(\rho,v)$ for all vertices in
	$V$.
	Therefore, as $\rho$ is a simultaneous fixed point of  all $\Lift(\cdot,v)$-operators
	and $\trho$ is the least such fixed point, 
	we obtain $\rho(v) \succeq \trho(v)$ for all $v \in V$. 
	Hence we have $\rho(v) = \trho(v)$ for all $v \in V$.
\end{proof}

The following lemmas prove each part of the invariant separately. The first part of the 
invariant describes the anti-monotonicity property which is needed to compute the $\best$
function with the $\CP_z$ operator. 

\begin{lemma}\label{lem:antimonoton}
	Invariant~\ref{inv:correctness}(1) holds:
	Let $r_1,r_2 \in \W$ and $r_1 \preceq r_2$.
	Before each iteration of the while-loop at Line~\ref{alg:bb3:while1} we have that if a vertex
	$v$ is in a set $S_{\succeq r_2}$ then it is also in $S_{\succeq r_1}$  (anti-monotonicity).
\end{lemma}
\begin{proof}
	We prove the claim by induction over the iterations of the while-loop.
	Initially, the claim is satisfied as the only non-empty set is $S_{min}$.
	It remains to show that when the claim is valid at the beginning of an
	iteration, then the claim also holds in the next iteration. 
	By induction hypothesis, the claim holds for the sets at the beginning of the while-loop. 
	In the trivial case, the algorithm terminates and the claim holds by the induction hypothesis. 
	Otherwise, the sets are only modified at Line~\ref{alg:bb3:liftvertices} and
	stored at Line~\ref{alg:bb3:store}.
	First, the vertices $P \cap V_c$ are added into the set $S_{\succeq r'}$. 
	Let $r''=D.getPrevious(r')$.
	Notice that after activating $r'$ all $r$ with $r'' \prec r \prec r'$ refer to the same set as $r'$
	and thus we add $P \cap V_c$ implicitly to all $r$.
	In the next iteration the while-loop then adds $P \cap V_c$ also to the set $S_{\succeq r''}$.
	As this done iteratively until a set $S_{\succeq r^*}$ with $P \cap V_c \subseteq S_{\succeq r^*}$ is reached (Lines~\ref{alg:bb3:while2}-\ref{alg:bb3:while2end}),
	the algorithm ensures that $P \cap V_c$ is contained in all set $S_{\succeq r''}$ with $r \prec r'$.
	By induction hypothesis we know that the invariant holds for all $r_2 \succ r'$ ($S_{\succeq r_2}$ is unchanged),
	and as the algorithm added $P \cap V_c$ to all set $S_{\succeq r''}$ with $r'' \preceq r'$
	the claim holds for all $r_1, r_2 \preceq r'$.  
\end{proof}

The second part of the invariant shows that the fixed point 
Algorithm~\ref{alg:blackbox3} computes is always smaller or equal to the least
fixed point. In particular, the fixed point computed by the algorithm is
defined as $\rho(v) = \max\{r \in \W \mid v \in S_{\succeq r}\}$ and we denote
the least fixed point with $\trho$. The proof is by induction: In the
beginning, every vertex is initialized with the minimum element which obviously
suffices for the claim. When we apply the $\lift$ function to vertices, we observe
that by the induction hypothesis the current value of a vertex is below or equal to
the fixed point. Additionally,  we obtain  a rank which is also smaller or equal to the lifted value of $\trho$ for every vertex 
as $\lift$ is a monotonic function.

\begin{lemma}\label{lem:underlfp}
	Invariant~\ref{inv:correctness}(2) holds: Throughout
	Algorithm~\ref{alg:blackbox3} we have $\trho(v) \geq \rho(v)$ for all $v \in V$.
\end{lemma}
\begin{proof}
	Before the while-loop at Line~\ref{alg:bb3:while1} the claim is obviously
	satisfied as $\trho(v) \succeq \min$ for all $v \in V$.
	We prove the claim by induction over the iterations of the while-loop:
	Assume we have $\rho(v) \preceq \trho(v)$ for all $v \in V$ before an iteration of
	the while-loop. 
	The function $\rho(\cdot)$ is only changed at
	Line~\ref{alg:bb3:liftvertices} and stored at Line~\ref{alg:bb3:store} where the set 
	$(P \cap V_c)$ is added to $S_{\succeq r'}$.
	For $v \in P \cap V_c$ we have that $v$ is a priority $c$ vertex and
	either $v$ is a player-$z$ vertex with a successor in $S_{\succeq r}$ or
	a player-$\z$ vertex with all successors in $S_{\succeq r}$.
	Thus, $r \preceq \best(\rho,v)$ for $v \in P$.
	At Line~\ref{alg:bb3:lift} we compute the $\lift$-operation
	for ranking $r$ with priority $c$ which results in the ranking $r'$ for the first iteration of the while-loop.
	By the monotonicity of the $\lift$ operation and the induction hypothesis we have that
	$r'=\lift(r,c)(v) \preceq \lift(\best(\rho,v),c)(v) \preceq \lift(\best(\trho,v),c)=\trho(v)$ for $v \in P \cap V_c$
	and thus adding $v$ to $S_{\geq r'}$ maintains the invariant
	(if $\rho(v) \succ r'$ beforehand it is not changed and otherwise it is lifted to $r' \prec \trho(v)$).
	In the later iterations of the while-loop $P \cap V_c$ is added to sets with smaller $r'$, which does not 
	affect $\rho$, as these vertices already appear in sets with larger rank.
\end{proof}

The following lemma proves the third part of Invariant~\ref{inv:correctness}:
Either there is an active $r \in \W$, i.e., the set $S_{\geq r}$ needs to be
processed, or $\rho(v)$ is a fixed point. 
We prove the property again by induction:
Initially, the set $\mina \in \W$ is active and every other set is empty
which trivially fulfills the property. Then, in every iteration
when we change a set with value $r$ we either activate it, or
there is a set with a value $r' \succeq r$ where $S_{\succeq r'}$ subsumes
$S_{\succeq r}$. In the former case, the condition is instantly
fulfilled. In the latter case, there is no vertex $v$ where $best(\rho,v)= r$
which renders $S_{\succeq r}$ irrelevant by definition of $\rho$.  

\begin{lemma}
	Invariant~\ref{inv:correctness}(3) holds: For all $r \in \W$:
	\begin{compactenum}
	\item $S_{\succeq r}$ is active or,
	\item $\forall v \in \CP_z(S_{\succeq r})\!:
		\text{ if } \best(\rho,v) = r$, then	$\rho(v) = \Lift(\rho,v)(v)$
	\end{compactenum} 
\end{lemma}
\begin{proof}
    
	We prove this invariant by induction over the iterations of the
	while-loop: Before the while-loop at Line~\ref{alg:bb3:while1} the claim is
	obviously satisfied as we activate $\mina$ which contains all
	vertices; for all other $r \in \W$ the set $S_{\succeq r}$ is empty and thus
	condition (2) is trivially satisfied.

	Assume the condition holds at the beginning of the loop.
	We can, therefore, assume by the induction hypothesis that the condition holds for all the sets.
	If there is no active $r \in \W$, the algorithm terminates and the condition
	holds by the induction hypothesis.
	The condition for a set $S_{\succeq r}$ can be violated only if either the set $S_{\succeq r}$ is changed or
	the set $S_{\succeq r}$ is deactivated.
	That is either at Line~\ref{alg:bb3:deactivate}, Line~\ref{alg:bb3:liftvertices} 
	or Line~\ref{alg:bb3:removed} of the algorithm.

	Let us first consider the changes made in the while-loop.
	If a set $S_{\succeq r}$ is changed in Line~\ref{alg:bb3:liftvertices}, then the algorithm
	either activates $r$ (Line~\ref{alg:bb3:activate}) and thus satisfies (1)
	or $S_{\succeq next(r)} = S_{\succeq r}$ which implies that $\best(\rho,v)
	\neq r$ and thus (2) is fulfilled trivially.
	At Line~\ref{alg:bb3:removed} there is no vertex $v$ with $\best(\rho,v) = r$ 
	(as there is no vertex $w$ with $\rho(w)=r$) 
	and thus (2) is satisfied (and it
	is safe to remove/deactivate the set in  Line~\ref{alg:bb3:removed}).

	Now consider the case where we remove the set $S_{\succeq r}$ 
	and make $r$ inactive at Line~\ref{alg:bb3:deactivate}.
	If the set $S_{\succeq r}$ is unchanged during the iteration of the outer while-loop then
	$S_{\succeq r}$ satisfies condition (2) after the iteration. This is because 
	for all $v$ with $\best(\rho,v) = r$ and $\alpha(v)=c$ we have that
	if $v$ is not already contained in $S_{\succeq \lift(r,c)}$ the algorithm adds it to the set $S_{\succeq \lift(r,c)}$
	in Line~\ref{alg:bb3:liftvertices} in the first iteration of the while-loop
	when processing $c$. This is equivalent to applying $\Lift(\rho,v)(v) =
	\lift(r,c)$. 
	If the set $S_{\succeq r}$ is changed during the iteration then this happens in the inner while-loop.
	As argued above, then either $r$ is activated and thus
	satisfies (1) or  $S_{\succeq next(r)} = S_{\succeq r}$ holds. Thus,
	there is no vertex $v$ with $\best(\rho,v) = r$, i.e., (2) is satisfied.
\end{proof}

\subsubsection{Symbolic Resources}

In the following, we discuss the amount of symbolic resources
Algorithm~\ref{alg:blackbox3} needs. We determine the number of \os
operations, the number of basic set operations and the
symbolic space consumption.

\begin{proposition}\label{prop:numberactivations}
	The number of \os operations in Algorithm~\ref{alg:blackbox3} is in $O(n \cdot |\W|)$.
\end{proposition}
\begin{proof}
	Each iteration of the while-loop at Line~\ref{alg:bb3:while1} processes an
	active $r$. That means, that the set $S_{\succeq r}$ was changed in a prior
	iteration. We use a symbolic one-step operation at Line~\ref{alg:bb3:cp} for
	each active $S_{\succeq r}$. It, therefore, suffices to count the number of possibly
	active sets throughout the execution of the algorithm. 
	Initially only $S_{\succeq \min}$ is active. After extracting an active
	set out of the data structure $D$, it is deactivated at Line~\ref{alg:bb3:deactivate}.
	We only activate a set $S_{\succeq x}$ when a new vertex is added to it at
	Line~\ref{alg:bb3:activate}. Because there can only be $n$ vertices with
	ranking $\succeq x$ for all $x \in \W$ the size of each set  $|S_{\succeq x}|$ is smaller or equal to $n$. 
	In the worst case, we eventually put every vertex into every set $S_{\succeq x}$ where $x \in
	\W$. Thus we activate $n \cdot |\W|$ sets which is equal to the number of 
	symbolic one-step operations.
\end{proof}

A similar argument works for analysing the number of basic set operations. 
\begin{proposition}\label{prop:bb3:setoperations}
	The number of basic set operations in Algorithm~\ref{alg:blackbox3} is in
	$O(\numprio \cdot n \cdot |\W| )$.
\end{proposition}
\begin{proof}[Proof of Proposition~\ref{prop:bb3:setoperations}.]
	As proven in Proposition~\ref{prop:numberactivations}, there are $O(n \cdot|\W|)$
	iterations of the outer while-loop and thus $O(n \cdot |\W|)$ iterations of the for-loop. 
	Thus the inner while-loop is started $O(\numprio n \cdot |\W|)$ times.
	The test whether the while-loop is started only requires two basic set operations and the overall costs 
	are bound by $O(\numprio n \cdot |\W|)$.
	We bound the overall costs for the iterations of the inner while-loop by an amortized analysis.
	First, notice that each iteration just requires 8 basic set operations 
	(including testing the while condition afterward).	
	In each iteration for a value $r' \in \W$ we charge the $r'$ for the
	involved basic set operations.
	Notice, that in each such an iteration new vertices are added to the set $S_{\succeq r'}$
	and thus $r'$ is processed at most $n$ times.
	Thus each $r' \in \W$ is charged for at most $8n$ basic set operations 
	Therefore, the number of basic set operations is $O(\numprio \cdot n \cdot |\W|) + O(n \cdot
	|\W|) = O(\numprio \cdot n \cdot |\W|)$. 
\end{proof}

Due to Proposition~\ref{lem:bb3:correctness},
Proposition~\ref{prop:numberactivations},
Proposition~\ref{prop:bb3:setoperations} and the fact that we use $\leq n$ sets
in the data structure $D$, we obtain Theorem~\ref{thm:bb3}.

\begin{theorem}\label{thm:bb3}
	Given a parity game, a finite total order $(\W,\succ)$ and a monotonic function $\lift$
	we can compute the least fixed point of all
	$\Lift(\cdot,v)$ operators with $O(n \cdot |\W|)$ \os operations, $O(\numprio \cdot n \cdot |\W|)$ basic set operations, 
	and $O(n)$ symbolic space. 
\end{theorem}

\section{Implementing the Ordered Progress Measure}

In this section, we plug the ordered approach to progress measure (OPM) described
by Fearnley et al.~\cite{FearnleyJS0W17} into Algorithm~\ref{alg:blackbox3}. To do this, 
we recall the witnesses they use in their algorithm and encode it with a specially-tailored
technique to obtain an algorithm with a sublinear amount of symbolic space. 
Finally, we argue that the function $\lift:\W \times C \mapsto \W$ and the total
order $(\W,\preceq)$ described in~\cite{FearnleyJS0W17} can be used to
fully implement Algorithm~\ref{alg:blackbox3}.

\smallskip\noindent\emph{The Ordered Progress Measure.}
To implement the ordered progress measure algorithm we need to argue that the
$\lift$-operation is monotonic in the first argument and the order
$(\W,\preceq)$ is a total finite order in order to fulfill the conditions of
Algorithm~\ref{alg:blackbox3}.
Let $\PG$ be a parity game and $C = \{0,\dots,\numprio-1\}$ be the set of priorities
in $\PG$. 
The set $\W$ in the ordered progress measure consists of tuples of priorities of length $k$,
where $k \in O(\log n)$. 
Each element in the tuple is an element of $C\_ = C \cup \{\_\}$, i.e., it is either a priority or "$\_$". 
The set $C\_$ has a total order $(C\_, \preceq)$ such that 
$\_$ is the smallest element, 
odd priorities are order descending and are considered smaller than 
even priorities which are ordered ascending.
The order $(\W,\preceq)$ is then obtained by extending the order $(C\_,\preceq)$
lexicographically to the tuples $r \in \W$.

For the details of the $\lift$ function we refer the reader to the work of
Fearnley et al.~\cite{FearnleyJS0W17}. An implementation of the lift
operation can be found at the GitHub repository
of the Oink system~\cite{dijk2018oink}.

By the results in \cite{FearnleyJS0W17} the order $(\W,\preceq)$ and $\lift$ meet the requirements
of our algorithms.

\begin{lemma}\label{lem:lift_monotonic}
   The following holds:
	(1) The function $\lift:\W \times C \mapsto \W$ is monotonic in the first
	parameter~\cite[p.6]{FearnleyJS0W17}.
	(2) The order $(\W,\preceq)$ is a total finite order~\cite[p.3]{FearnleyJS0W17}.
	(3) Let $\rho$ be the least simultaneous fixed point of all $\Lift(\cdot,v )$ operators. Then 
	$\rho(v) = \top$ iff player $\E$ has a strategy to win the parity game $\PG$
	when starting from $v$~\cite[Lemma 7.3, Lemma 7.4]{FearnleyJS0W17}.
\end{lemma}

Theorem~\ref{thm:bb3} together with
Lemma~\ref{lem:lift_monotonic} imply the following theorem.

\begin{theorem}\label{thm:ordered_main}
	Algorithm~\ref{alg:blackbox3} implemented with the OPM computes the winning set
	of a parity game with $O(n\cdot|\W|)$ \os operations, $O(\numprio \cdot n \cdot |\W|)$
	basic set operations, and $O(n)$ symbolic space.
\end{theorem}

\begin{remarkb}{(Bounds for $|\W|$).}\label{rem:bound}
	We now discuss the bounds on $|\W|$. The breakthrough result of~\cite{calude2017stoc} 
	shows that $|\W|$ is quasi-polynomial ($n^{O(\log \numprio)}$) in general and polynomial
	when $\numprio \leq \log n$. Using the refined analysis of~\cite{FearnleyJS0W17}, 
	we obtain the following bound on $|\W|$: 
	in general, $\min(n \cdot \log(n)^{\numprio-1}, h \cdot
	n^{c_{1.45}+\log_2(h)})$, where $c_{1.45} = \log_2(e) < 1.45$ and $h=\lceil 1+\numprio /\log(n) \rceil$;
	and if $\numprio \leq \log n$, then $|\W|$ is polynomial due
	to~\cite[Theorem 2.8]{calude2017stoc} and \cite[Corollary 8.8]{FearnleyJS0W17}.
	Note that  $O\left(n^{2.45+\log_2(\numprio )}\right)$ gives a naive upper bound on $|\W|$ in general. 
	Plugging the bounds in Theorem~\ref{thm:ordered_main} we obtain a set-based symbolic 
	algorithm that requires quasi-polynomially many \os and basic set operations and $O(n)$ symbolic space.
	The algorithm requires only polynomially many \os and basic set operations when $\numprio  \leq \log n$.
\end{remarkb}

\section{Reducing the Number of Sets for the OPM}\label{sec:reducesets}

In this section, we tailor a data structure for the OPM  in order to only use
$O(\numprio \cdot \log n)$ sets. While each progress measure can be encoded by $\log(|\W|)$
many sets, the challenge is to provide a representation that also allows to efficiently compute the sets $S_{\succeq r}$.
Such a representation has been provided for the small progress measure~\cite{chatterjee2017symbolic} 
and in the following we adapt their techniques for the OPM.

\smallskip\noindent\emph{Key Idea.}
The key idea of the symbolic space reduction is that we \emph{encode} the
value of each coordinate of the rank $r$ separately. A set no longer just
stores the vertices with specific rank $r=b_1 \dots b_k$ but instead stores
all vertices where, say, the first coordinate $b_1$ is equal to a specific value in
$C\_$. This encoding enables us to use only a polylogarithmic amount of symbolic space
under the assumption that the number of priorities in the game graph is
polylogarithmic in the number of vertices.

\smallskip\noindent\emph{Symbolic Space Reduction.}
Let the rank of $v$ be $r=b_{1}\dots b_k$.
Vertex $v$ is in the set $C_x^i$ iff the $i$th coordinate of the rank of $v$ is $x$
and a vertex $v$ is in the set $C_\top$ iff the rank of $v$ is $\top$. 
Thus $O(\log(n) \cdot \numprio)$ sets suffice to encode all $r \in
\W$. We demonstrate this encoding of the sets in
Example~\ref{ex:setex}.

\begin{example}\label{ex:setex}
	Let $\PG$ be a parity game containing the vertices $v_1,v_2,v_3$. Assume the
	following ranking function:
	$
		f(v_1)  = 65433, f(v_2) = 75422, f(v_3) =
		\mathunderscore\mathunderscore\mathunderscore 32.
	$
	Using the definition of our encoding, we have that:
	$
		\{v_3\} \subseteq C_{\_}^1, \{v_1\} \subseteq C_6^1,\{v_2\} \subseteq C_7^1, 
		\{v_3\} \subseteq C_{\_}^2,	\{v_1,v_2\} \subseteq C_5^2,\{v_3\} \subseteq C_{\_}^2, 
		\{v_1,v_2\} \subseteq C_4^3,\{v_2\} \subseteq C_2^4, \{v_1,v_3\} \subseteq C_3^4, \{v_1\} \subseteq C_3^5,\{v_2, v_3\} \subseteq C_2^5
	$.
\end{example}

\smallskip\noindent\emph{Computing the set $S_{\succeq r}$ from $C_x^i$.}
We obtain the set $S_r$ for rank  $r=b_{1}\dots b_k$ with an intersection of the sets 
$\bigcap_{i=1}^k C_{b_i}^i = S_r$.
To acquire the set $S_{\succeq r}$ we first consider
sets where the first $i$ elements are the equal to $b_1, \dots b_i$
but the $i\mathtt{+}1$th element $x$ is $\succ b_{i+1}$.	
\begin{equation}\label{eq:constructsgreaterr}
  S^i_{\succeq r} =\bigcap_{1 \leq j \leq i} C_{b_j}^j \cap \bigcup_{x\succ b_{i+1}} C_{x}^{i+1}
\end{equation}

To \emph{construct} the set $S_{\succeq r}$ we apply the following union
operations:
\begin{equation}\label{eq:constructsets}
	S_{\succeq r} = \bigcup_{i = 1}^{k-1} S^i_{\succeq r} \cup S_r \cup C_{\top}
\end{equation}

That is, we can compute the set  $S_{\succeq r}$ with $O(\numprio \cdot \log n)$ set operations and four additional sets.
Notice that there is no need to store all sets $S^i_{\succeq r}$ as we can immediately add them to the final set when we have computed them.
The number of $\cup$-operations is immediately bounded by $O(\numprio \cdot k)=O(\numprio \cdot \log n)$ by the above definitions.
In order to bound the number of $\cap$-operations by $O(\log n)$, we do the following.
To compute the sets $S^i_{\succeq r}$  we introduce an additional set 
$T^{i} = \bigcap_{1 \leq j \leq i} C_{b_j}^j$.
We have that $S^i_{\succeq r} = T^i \cap \bigcup_{x\succ b_{i+1}} C_{x}^{i+1}$
and $T^{i+1}= T^{i} \cap C_{b_{i+1}}^{i+1}$, i.e., we just need two $\cap$ operation to compute the next set $S^{i+1}_{\succeq r}$.
Moreover, we have that $S_r=T^k$ and thus can be computed with just one $\cap$ operation.
In total, this amounts to $2k-2 = O(\log n)$ many  $\cap$-operations.

\smallskip\noindent\emph{Updating the set $S_{\succeq r}$ to $S'_{\succeq
		r}$.}\label{par:updatesets}
Assume that the set $S_{\succeq r}$ is the old set that is saved within the sets
$C^i_x$. The new set, $S'_{\succeq r}$ is an updated set, which is also a
superset.
First, compute the difference $S_\Delta = S'_{\succeq r} \setminus S_{\succeq r}$.
Intuitively, the algorithm increased the rank of the vertices in $S_\Delta$.  We delete their old values by updating 
$C^i_{x} = C^i_{x} \setminus S_\Delta$ \ for all $i = 0\dots k$ and each $x \in
\{0,\dots \numprio-1\}$. Then we add the vertices to the set  $C^i_{r_i} = C^i_{r_i}
\cup S_\Delta$ for all $i = 0\dots k$.
In total there are $O(\numprio \log n)$ many $\setminus$-operations and $O(k) =
O(\log n)$ many $\cup$-operations.

Using the above techniques for constructing and updating the sets
Algorithm~\ref{alg:blackbox3} can be modified to obtain the following theorem.
We present the details in Section~\ref{app:detailsSetAlg}.

\begin{theorem}\label{thm:bb4}\label{thm:reducedspace_main}
	The winning set of a parity game can be computed in  $O(n \cdot |\W|)$ \os operations,
	$O(\numprio^2 n \cdot |\W| \cdot \log n)$ basic set operations, and $O(\numprio \cdot \log n)$ symbolic space. 
\end{theorem}

\begin{remarkb}\label{remark:polyspace}
	Note that Theorem~\ref{thm:bb4} achieves bounds similar to Theorem~\ref{thm:bb3}
	with a factor $\numprio \cdot \log n$ increase in basic set operations,
	however, the symbolic space requirement decreases from $O(n)$ to $O(\numprio \cdot \log n)$. 
	In particular, using the bounds as mentioned in Remark~\ref{rem:bound}, we obtain 
	a set-based symbolic algorithm that requires quasi-polynomially many \os and
	basic set operations, and $O(\numprio \cdot \log n)$ symbolic space,
	and moreover, when $\numprio \leq \log n$, then the algorithm requires polynomially
	many \os and basic set operations and 
	only poly-logarithmic $O(\log^2 n)$ symbolic space. 
\end{remarkb}

\begin{remarkb}\label{remark2}
	Recall that our AVL-tree data structure potentially requires $O(n)$ non-symbolic space 
	(cf.~Remark~\ref{remark1}) in the worst case. 
	The algorithm above reduces the symbolic space requirement to $O(\numprio \cdot \log n)$.
	We briefly outline how to reduce the non-symbolic space to the same bound.
	The main purpose of our AVL-tree data structure is to (a) avoid storing all sets explicitly and
	(b) efficiently maintain pointers to the active sets. 
	As we now have a succinct representation of all sets we are only left with (b).
	For (b), we additionally maintain $d\cdot \log n +1$ sets in a data structure $O$ to determine if the set
	$S_{\succeq r}$ was changed since the last time we chose $r$ at the start of
	the while-loop which processes an active $r$ in each iteration. 
	In each iteration of the outer while loop the chosen active
	set $S_{\succeq r}$ of $D$ is copied into the corresponding set $T_{\succeq r}$ in $O$.
	We can then identify sets $S_{\succeq r}$ which changed since the last iteration where $r$ was chosen by testing
	$(S_{\succeq r} \neq T_{\succeq r})$. 
	To obtain an active element, it remains to additionally check whether there is no $r' >
	r$ where $S_{\succeq r'}  = S_{\succeq r}$. If $S_r \neq \emptyset$ this
	property is true and we can return an active set. Because we need to
	possibly go through $\W$ in each iteration, the number of basic set operations
	is increased by a factor of $\W$. 
	In total, the number of basic set operations are increased 
	to $O(\numprio \cdot n \cdot |\W|^2 \log n)$ while
	the other symbolic resource consumption stays the same.
\end{remarkb}

\subsection{Proof of Theorem~\ref{thm:reducedspace_main}}\label{app:detailsSetAlg}

\begin{algorithm}[ht]
	\SetKwInOut{Input}{input}
	\SetKwInOut{Output}{output}
	\SetAlgoVlined
	\Input{Parity Game $\PG$}
	\caption{OPM Algorithm with Reduced Symbolic Space}\label{alg:bb4}
	Initialize $C^i_{\_} \gets  V$ for $0 \leq i \leq k$\; \label{alg:bb4Linit1}
	Initialize $C^i_{c} \gets \emptyset$ for $0 \leq i \leq k$, $c \in C$\; \label{alg:bb4Linit2}
	D.activate$(\min)$\;
	\While{$r \gets D.popActiveSet()$}{\label{alg:bb4:while1}
		$S_{\succeq r}  \gets D.getSet(r)$\;
		D.deactivate($r$)\;
		$P \gets \CP_z(S_{\succeq r})$\;\label{alg:bb4:getCP}
		\For{$c \in C$} {\label{alg:bb4:for1}
			$r' \gets \lift(r,c)$\;
			$S_{\succeq r'}  \gets D.getSet(r')\;$\;\label{alg:bb4:getSet1}
			$rold \gets r'$\;\label{alg:bb4:roldv}
			$S_{rold} \gets S_{\succeq r'} \cup (P \cap
			V_c)$\label{alg:bb4:rolds}\;
			\While{$P \cap V_c \not\subseteq S_{\succeq
					r'}$}{\label{alg:bb4:while2}
				$S_{\succeq r'} \gets S_{\succeq r'} \cup (P \cap V_c)$\;
				$S_{\succeq next(r')} \gets D.getSet(D.next(r'))\;$\;
				\If{$r' = \top $ or $S_{\succeq r'} \supset S_{\succeq next(r')}$}
				{
					$D.activate(r', S_{\succeq r'})$\;
				}
				$S_{\succeq prev(r')} \gets D.getSet(D.getPrevious(r'))$\;
				\If{$S_{\succeq r'} = S_{\succeq prev(r')} \cup (P \cap V_c)$}{
					$D.removeSet(D.getPrevious(r'))$\;
				}				
				$r'\gets D.getPrevious(r')$\;
				$S_{\succeq r'}  \gets D.getSet(r')\;$\;
			}   
			$D.update(rold, S_{rold})\;$\;
			\label{alg:bb4:update}
		}
	}
	\Return $S_\top$
\end{algorithm}

\smallskip\noindent\emph{Data Structure.}\label{alg:bb4:ds}
The new data structure $D$ supports all functions of Data
Structure~\ref{alg:bb3:ds} but the nodes of the AVL tree no longer store a
pointer to a set (but only the value $r$).
The data structure $D$ stores $\numprio \cdot \log n + 1 $ sets as described in Section~\ref{sec:reducesets}, and 
overrides the $D.update(r,S)$-operation.

\begin{itemize}
	\item D.update(r,S): Updates the set $S_{\succeq r}$ to be $S$.
		Every set $S_{\succeq r'}$ with $r' \prec r$  is updated to
		$S_{\succeq r'}\cup S$.
\end{itemize}

\smallskip\noindent\emph{Implementation of the data structure $D$.}
\begin{itemize}
	\item D.getSet(r): This function computes the set $S_{\succeq r}$ as
		described in Section~\ref{sec:reducesets} and returns it.
	\item D.update(r,S): Computes the update of the set $S_{\succeq r}$ with $S$ as
		described in Section~\ref{sec:reducesets}.
		Precondition: $S_{\succeq r} \subseteq S$.

\end{itemize}

Notice that Algorithm~\ref{alg:bb4} only differs from Algorithm~\ref{alg:blackbox3} in 
(a) the way the  \texttt{D.getSet(r)} method is implemented and 
(b) how the sets are updated, i.e., the overridden \texttt{D.update(r,S)} method.
Hence, to establish the correctness of Algorithm~\ref{alg:bb4}, it suffices to show that these two methods
do exactly the same as the corresponding operations in Algorithm~\ref{alg:blackbox3}.
The correctness then directly follows from Proposition~\ref{lem:bb3:correctness}.

\begin{proposition}[Correctness]\label{prop:bb4:correctness}
	Given a parity game $\PG$ Algorithm~\ref{alg:bb4}
	computes the winning set.
\end{proposition}

\begin{proof}
	We show that the \texttt{D.getSet(r)} method  (in the interplay with the \texttt{D.update(r,S)} method)
	in Algorithm~\ref{alg:bb4} returns the same sets as the \texttt{D.getSet(r)} method in 
	Algorithm~\ref{alg:blackbox3}. The correctness then directly follows from Proposition~\ref{lem:bb3:correctness}.

	The proof is by induction. Consider the base case after the initialization of the data structure $D$
	and its sets $C_c^i$.
	In Algorithm~\ref{alg:blackbox3} we have that $D.getSet(\min)$ would return $V$
	and $D.getSet(r)$ would return the empty set for $\min \prec r$.
	In Algorithm~\ref{alg:bb4}, by the initialization in line \ref{alg:bb4Linit1},
	we have that $D.getSet(\min)$ would return $V$ and 
	by the initialization in line \ref{alg:bb4Linit2} we have $D.getSet(r)$ would return the empty set for $\min \prec r$.
	Therefore, the base case is satisfied.

	Notice that the data structures both in Algorithm~\ref{alg:blackbox3} and
	Algorithm~\ref{alg:bb4} are only changed in the for-loop. 
	Assume the claim holds before an iteration of the for-loop. Let $\bar{r}$ be
	the element of $\W$ currently processed by the outer while-loop,
	and let $\bar{r}'$ the $r'$ currently processed by the for-loop.
	Consider some $r \in \W$. 
	By induction hypothesis, $D.getSet(r)$ coincides in both algorithms beforehand.
	If $r \succ \bar{r}'$ then $D.getSet(r)$ is not affected by the changes in both algorithms and 
	thus $D.getSet(r)$ coincides in both algorithms after the iteration of the for loop.
	If $r \preceq \bar{r}'$ then Algorithm~\ref{alg:blackbox3} updates the data structure such
	that $P \cup V_c$ is added to the set $S_{\succeq r}$ (the set returned by $D.getSet(r)$).
	Now consider Algorithm~\ref{alg:bb4}. Here the algorithm adds the set $P \cup
	V_c$ to the set $C_x^i$
	that correspond to $\bar{r}'$.
	As $r \preceq \bar{r}'$ there is an $i\geq 0$ such that $r$ and $\bar{r}'$ coincide on the first $i$ elements 
	and the set $P \cup V_c$ is then contained in the set $S^i_{\succeq r}$.
	That means that the set returned by $D.getSet(r)$ contains the set $P \cup V_c$. 
	Moreover, as only vertices in $P \cup V_c$ are affected by the update, all the vertices that were previously contained in
	$D.getSet(r)$ are still contained in the set. In other words, Algorithm~\ref{alg:bb4} adds $P \cup V_c$  to the set $S_{\succeq r}$.
	That is, the two $D.getSet(r)$ methods coincide also after the iteration of the for-loop for all $r \in \W$.
	Thus, we have that $D.getSet(r)$ coincides in the two algorithms and the correctness of Algorithm~\ref{alg:blackbox3}
	extends to Algorithm~\ref{alg:bb4}.
\end{proof}

\begin{proposition}[Symbolic operations]\label{prop:bb4:runningtime}
	Algorithm~\ref{alg:bb4} uses
	$O(\numprio \log n)$ sets with $O(n |\W|)$ symbolic one-step operations and
	$O( \numprio^2 n |\W| \log n)$ basic
	set operations.
\end{proposition}
\begin{proof}
	There are $O(n \cdot |\W|)$ iterations if the while-loop at Line~\ref{alg:bb4:while1} (cf.~Proposition~\ref{prop:numberactivations}). 
	Therefore, the number of \os operations is $O(n \cdot |\W|)$.
	In each iteration of the while-loop, the
	for-loop at Line~\ref{alg:bb4:for1} has $\numprio$ iterations. 
	The basic set operations at Line~\ref{alg:bb4:getSet1},
	Line~\ref{alg:bb4:while2}, Line~\ref{alg:bb4:rolds} and at
	Line~\ref{alg:bb4:update} occur $O(\numprio \cdot n \cdot
	|\W|)$ times. This sums up to a total of $O(\numprio^2 \cdot n \cdot |\W| \log n)$
	basic set operations (as each \texttt{getSet(r)} requires $O(\numprio \log n)$
	basic set operations). 
	By the same amortized argument as in the proof of
	Proposition~\ref{prop:numberactivations},
	we obtain that the total number of basic set-operations in executions of the inner while-loop is in 
	$O(n |\W| \numprio \log n)$.
	The number of basic set operations is, thus, in $O(\numprio^2 n |\W| \log
	n)+ O(\numprio n |\W| \log n) = O(\numprio^2  n \cdot |\W|  \log n)$. 
\end{proof}

\noindent Due to Proposition~\ref{prop:bb4:correctness} and
Proposition~\ref{prop:bb4:runningtime} and the fact that we use only $O(\numprio \log
n)$ symbolic space in the modified data structure $D$, we obtain Theorem~\ref{thm:bb4}.

\section{Conclusion}

In this work, we present improved set-based symbolic algorithms for parity games.
There are several interesting directions for future work. 
On the practical side, implementations and experiments with case studies,
especially for the algorithm presented in Section~\ref{sec:bb3} instantiated with either the ordered approach or the succinct progress measure,
is an interesting direction.
On the theoretical side, recent work~\cite{ChatterjeeDHL18} has established lower bounds for 
symbolic algorithms for graphs, and whether lower bounds can be established for 
symbolic algorithms for parity games is another interesting direction for future work.
\bibliography{mpg}
\end{document}